\documentclass[letterpaper, 10 pt, conference]{ieeeconf}  

\usepackage{color}

\usepackage{amsmath} 
\usepackage{amssymb}  

\usepackage{todonotes}

\usepackage{algorithm}
\usepackage{balance}
\usepackage{algpseudocode}
\algtext*{EndIf}
\usepackage{graphicx}
\usepackage{capt-of}
\makeatletter
\def\BState{\State\hskip-\ALG@thistlm}
\makeatother

\usepackage{psfrag}
\usepackage{subcaption}
\usepackage[font=small,labelfont=bf]{caption}

\usepackage{tikz}
\usetikzlibrary{shapes,arrows,positioning}

\newtheorem{lemma}{Lemma}
\newtheorem{proposition}{Proposition}

\newtheorem{assumption}{Assumption}

\usepackage[compress]{cite}

\IEEEoverridecommandlockouts                              

\title{\LARGE \bf
A Spiking Neural Dynamical Drift-Diffusion Model on Collective Decision Making with Self-Organized Criticality
}

\author{ Yanlin Zhou, Chen Peng, and Qing Hui
\thanks{This work was supported by the Defense Threat Reduction Agency, Basic Research Award \#HDTRA1-15-1-0070.
}
\thanks{The authors are with the Department of Electrical and Computer Engineering, University of Nebraska-Lincoln, Lincoln, NE 68588-0511, USA, {\tt\small yanlin.zhou@huskers.unl.edu; chen.peng@huskers.unl.edu;  qing.hui@unl.edu}.}%
}

\begin{document}

\maketitle
\thispagestyle{empty}
\pagestyle{empty}

\begin{abstract}

This article proposes a novel collective decision making scheme to solve the multi-agent drift-diffusion-model problem with the help of spiking neural networks. The  exponential integrate-and-fire model is used here to capture the individual dynamics of each agent in the system, and we name this new model as Exponential Decision Making (EDM) model. 
We demonstrate analytically and experimentally that the gating variable for instantaneous activation follows Boltzmann probability distribution, and the collective system reaches meta-stable critical states under the Markov chain premises. With mean field analysis, we derive the global criticality from local dynamics and achieve a power law distribution. Critical behavior of EDM exhibits the convergence dynamics of Boltzmann distribution, and we conclude that the EDM model inherits the property of self-organized criticality, that the system will eventually evolve toward criticality. 

\end{abstract}

\section{Introduction}\label{sec:intro}

Self-Organized Criticality (SOC), a ground breaking achievement of statistical physics, has gained growing interest in neural firing and brain activity in recent years \cite{soc_fundm}. Bak's hypothesis \cite{bak_1996} and recent studies \cite{soc_fundm,Neurobiologically_SOC_Rubinov} suggest that criticality is evolutionarily chosen for optimal computational and fast reactionary purposes, and that the brain is always balanced precariously at the critical point.

Such critical dynamics emerge during the phase transition between randomness (sub-critical) and order (super-critical), and usually follow power-law distributed spatial and temporal properties \cite{Neurobiologically_SOC_Rubinov}. A dynamic network system with the SOC behavior then has the potential to be spatial and/or temporal scale free \cite{Watkins2016} and to fast switch between phases and attain optimal computational capability. This offers a possible approach to model a decision making process. 

The systems that exhibit SOC behavior are usually high dimensional and slowly driven, with nonlinearity properties \cite{bak_1996,Watkins2016}. To this end, Brochini \textit{et al} \cite{phase_2016} have discussed the phase transitions and SOC in stochastic spiking neural networks.
Also, Bogacz \textit{et al} \cite{formal_ddm} demonstrated that standard Drift Diffusion Model (DDM) can be used for stochastic spiking dynamics and they relate DDM to a highly interactive ``pooled inhibition" model.

However, to authors' best knowledge, although the SOC has been recognized as a fundamental property of neural systems \cite{soc_fundm}, there has yet to be a decision making model capitalizing the SOC property.

In this paper, we incorporate the nonlinearity of Exponential Integrate-and-Fire (EIF) model to replace the stochastic spiking scheme in DDM proposed by \cite{formal_ddm}. Introduced by Fourcaud-Trocme \textit{et al} \cite{eif_fourcaud}, the nonlinear EIF model is experimentally verified to be able to accurately capture the response properties. It will be demonstrated that neural sampling and mean field branching can be derived with the Boltzmann distribution. The proposed Exponential Decision Making (EDM) model therefore reaches a set of absorbing states, and the corresponding global criticality follows power law distribution, thus attaining the SOC behavior.

The contributions of this paper are summarized as follows.
\begin{enumerate}
	\item To the best of authors' knowledge, this paper is the first published work on modeling decision making processes with the SOC property.
	\item A collective decision making model, i.e., EDM, is proposed to implement the DDM methodology on EIF spiking neurons.
	\item A probability inference scheme on EIF sampling is proposed, which extends an existing leaky integrate-and-fire sampling method.
	\item Mean field analysis of the connectivity of EDM is given, which exhibits global criticality.
	\item Detailed analysis is given to reveal SOC behavior of the EDM model under criticality conditions.
\end{enumerate}

This paper is organized as follows. 
In Section~\ref{sec:model}, literature review is given and the necessary background concepts are introduced.
Then, in Section~\ref{sec:dm_dynam}, the dynamics of each agent as well as the EIF neural sampling are discussed.
In Section~\ref{sec:coll_beh}, the collective behavior of the network system is analyzed.
Section~\ref{sec:converge} provides convergence analysis and simulation results.
Section~\ref{sec:conclusions} concludes the paper and looks into future work.

\begin{figure}[htbp]
  \centering
      \includegraphics[width=0.48\textwidth]{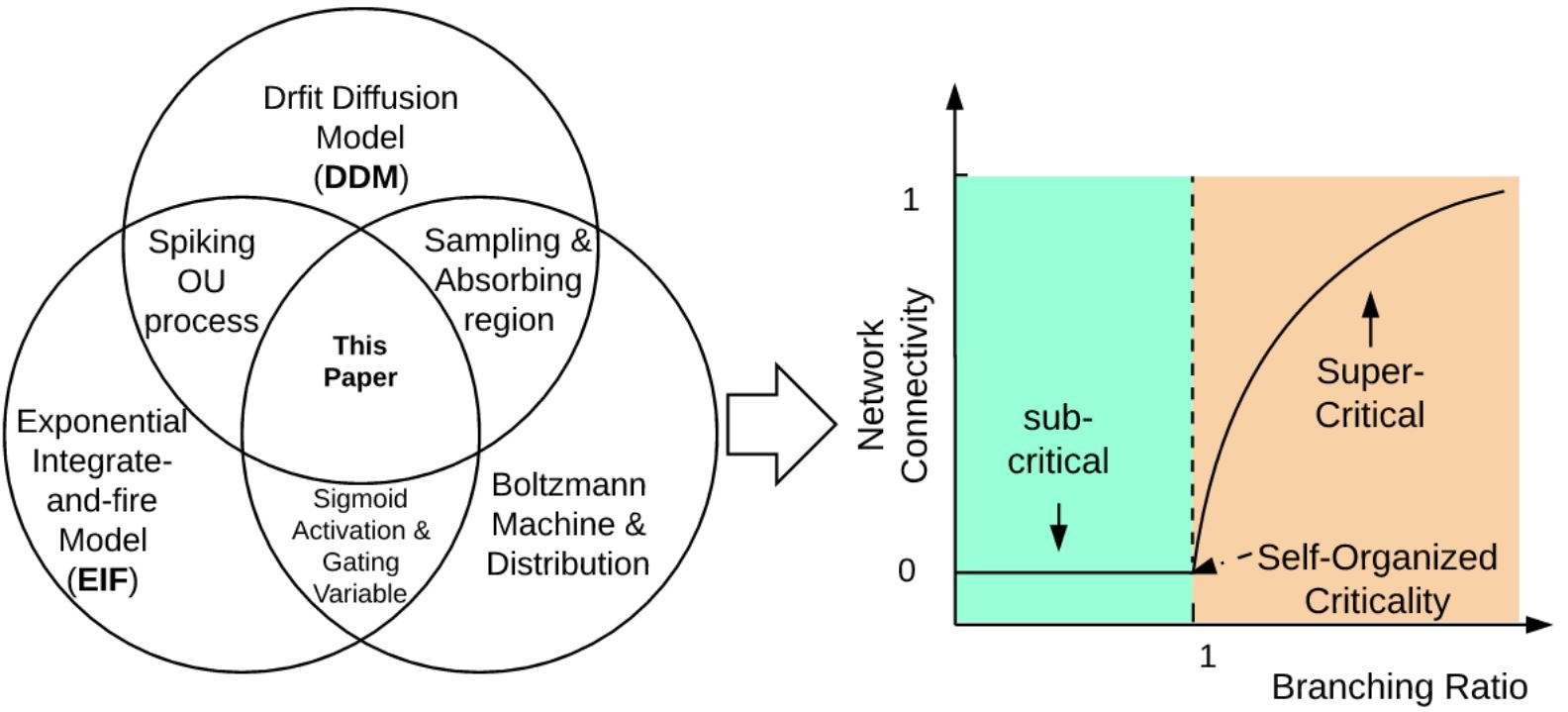}
  \caption{EDM model shows SOC property}{\label{fig:relation}}
\end{figure}

\subsection*{Notation and Preliminaries}

Here we use a classical directed graph representation $\mathcal{G}=(\mathcal{V},\mathcal{E},\mathcal{A})$ with nonempty finite number of nodes and edges.
Specifically, $\mathcal{V}$ is the set of nodes, $\mathcal{E}$ is the set of directed edges, and $\mathcal{A}=[a_{ij}]$ is an adjacency matrix with weights $a_{ij} > 0$ if, $\left<i,j\right>\in\mathcal{E}$, an edge from node $i$ to node $j$. 
Note that the assumed graph is simple, i.e., $\left<i,i\right>\notin\mathcal{E}$, $\forall i$, with no multiple edges going in the same direction between the same pair of nodes and no self loop. 
In this case, the diagonal elements of $\mathcal{A}$ are zero. In addition, the Laplacian matrix of $\mathcal{G}$ is denoted by $L$.

\section{Backgrounds and Related Information}\label{sec:model}

\subsection{Self-Organized Criticality}\label{ssec:socintro}

Self-Organized Criticality describes a self-tuned internal interactions that show critical dynamics in complex systems \cite{bak_1996}. The interacting node groups are called active sites while the nodes that are less sensitive to the input are called inactive sites. Usually in the sand pile model, each agent has their own steep slope, which represents the membrane potential of the spiking neurons. When a certain threshold is hit and the sand in that specific area is steep enough $Z_{local} > Z_{critical}$, the avalanches will be triggered, which follow a power-law distribution of $1/f$ noise.
Plenz and Beggs \cite{Beggs11167} observed a similar pattern of avalanches in the cortical neural electrical activity, which was the first evidence that the brain functions at criticality.

For the spiking neural network sense, if an agents activates too many neighboring neurons (super-critical), it leads to the massive activation of the entire network, while if too few neurons are activated (sub-critical), propagation dies out too fast \cite{formal_ddm}.

In our case, the local rigidity level is expressed by two terms: the firing threshold of each agent, and the correlation between each two agents in the same local active sites.

\subsection{Boltzmann Machine}\label{ssec:bmintro}

Boltzmann Machines (BM) is a special type of stochastic recurrent neural networks based on non-stochastic Hopfield nets. In recent years, BM's property of binary output attracts more attention in both the theoretical neuroscience and the high dimensional parallel stochastic computation \cite{bm_basic,neural_sampling_2011}. 

BM is proven to be efficient for the models with connectivity properly constrained, to be specific, machine learning and probability inference are two major applications. In this paper, we apply neural sampling and show that the probability function of the gating current variable for the activation term follows the Boltzmann functions.

The global energy function of the Boltzmann Machine is defined as:
$
		E = -  \sum_{i<j}h_{ij} s_i s_j - \sum_{i}b_i s_i 
$
, where 
$E$ is the global energy,
$h_{ij}$ is the connection strength between Unit $j$ and Unit $i$,
$s_i\in \{0,1\}$ is the state of Unit $i$, and
$b_i$ is the bias of Unit $i$.

As for the Boltzmann distribution, the probability that the $i$th unit is on is
	\begin{equation}\label{eq:bmprobdist}
		P_{i=on} = \frac{1}{1+\exp(-\triangle E_i/T)},
	\end{equation}
where $T$ is the temperature of the system.

The probability is calculated using only the information of the energy difference from the initial state, and the temperature of the current time. In our framework, we consider the term $-\triangle E_i/T$ as a logistic activation function as in \cite{Barranca2014,richardson_2009}. It has already been shown in \cite{richardson_2009} that some stochastic neurons sample from a Boltzmann distribution. Ideally, after a long running period and without further inputs,  the probability of a global state will not be affected by other terms, i.e., time constants and conductance values in the EIF model. At this stage, the system is at its ``thermal equilibrium", that converges to a low temperature distribution where the energy level hovers around a global minimum.

This feature presents a similar behavior as SOC if we consider the criticality as this thermal equilibrium that energy level fluctuates around. Also, the log-probabilities eventually becomes a linear term, which helps us simplify the exponential term in the EIF model. Further discussion will be given in later sections.

Moreover, the neural sampling technique in the later section incorporates the Boltzmann machine according to some local switching, with conditional probability integrated, the multivariable Boltzmann joint distribution has the form \cite{noise_resource}
	\begin{equation}\label{eq:boltzmenergyt}
		P_m=\dfrac{e^{-\epsilon_m/kT}}{\sum_{n=1}^M e^{-\epsilon_j/kT}},
	\end{equation}
where $P_m$ stands for the probability of state $m$, 
$\epsilon_m$ is the energy at state $m$, 
$k$ is a constant, and 
$M$ is the total number of the states.

\subsection{DDM}\label{ssec:ddm}
Drift Diffusion Model has been applied on Two-Alternative Forced-Choice (TAFC) task in an extensive amount of work (see \cite{formal_ddm,ddm_naomi} for instance).
The fact that DDM integrates the difference between two choices according to one or two threshold makes it possible to describe a decision making process in a spiking neural network.

In the pure DDM, the accumulation of the unbiased evidence has the form
	\begin{equation}\label{eq:pureddm}
		dx = g dt + \beta dw , ~x(0) = 0,
	\end{equation}
where
$dx$ represents the changes in difference over the time interval $dt$,
$g$ is the increase in evidence supporting the correct choice each time,
$w$ is the independent, identically distributed (i.i.d.) Wiener processes, and 
$\beta$ is the standard deviation.
The probability density $P(x,t)$ is normally distributed with mean $g t$ and standard deviation $\beta\sqrt{t}$.

Since the second term in (\ref{eq:pureddm}) is represented by a standard Wiener process that describes the noise, it is common to consider $dx$ in DDM to be the change in membrane potential within a certain amount of time \cite{stochastic_diffusion}.

Here we consider a network system with $N$ agents. Each agent relates the DDM model to the non-stationary dynamics of the firing of an EIF spiking neural model. While the forced-response protocol is usually considered, we follow the free-response protocol, that each consecutive fires determine the range of the time interval. The common assumption made for this equation usually considers $g>0$ to support the first choice, and $g<0$ for the other \cite{formal_ddm}. The term $g$ can either be a constant for inactive nodes, or a function for active nodes that depends on membrane potential. 

While (\ref{eq:pureddm}) only describes the dynamics of a single DDM system, we need extra terms to capture the impact from neighbors. We have the following stochastic diffusion process with an initial condition $x_0$:
	\begin{equation}\label{eq:voltage_ddm}
		dx = \big( \alpha(x(t),t)(x(t)-x_0) +g(x(t),t) \big) dt + \beta(x(t),t) dw,
	\end{equation}
where $\alpha(t)$ is a measurable gain function that models the external input to accelerate the potential increment, and linear drifting term $g(t)$ represents the dynamic drifting variable of the node itself.
In this regime, we have transferred our model to a Ornstein-Uhlenbeck (OU) process, which is known to be the solution to the famous Fokker-Planck equation. 
Here, to further simplify the model, we may eliminate the afterhyperpolarization, that is, let $x_0 = 0$.
	
For the model proposed above, it is possible to receive the spike generations with arbitrary shape, i.e., different spiking time intervals and different incremental speed of membrane potential. With a proper defined activation function, which will be discussed in Section \ref{sec:dm_dynam}, the behavior of each single stochastic diffusion process can be bounded. Before further discussing into the individual dynamics and their boundedness properties, we need to look into the specific local dynamics by applying a most commonly used neuron model.

\subsection{Generalized Exponential Integrate-and-Fire (EIF) Model}\label{ssec:EIF}

Exponential integrate-and-fire (EIF) is a well developed biological neuron model introduced by Fourcaud-Trocme \textit{et al} \cite{eif_fourcaud} as an extension of the standard leaky integrate-and-fire model. As concluded in several studies \cite{Barranca2014,richardson_2009}, EIF is a suitable simple model for very large scale network simulations. For the generalized EIF \cite{richardson_2009}, arbitrary spike shapes are allowed and gated currents usually reach a steady-state with nonlinear voltage activation function.

The EIF model holds a nonlinearity property consisting of a linear leakage term combined with an exponential activation term, which follows a simple RC-circuit dynamics before $V$, the membrane potential, reaches a set threshold $V_T$. After reaching the threshold, it can be considered that the neuron has fired, and its membrane potential is then set to a resting voltage, $V_R$, approximately $-60$ mV \cite{Barranca2014,richardson_2009} or $0$ mV \cite{stochastic_diffusion} by different assumptions.   

The dynamics of the membrane potential is given by
	\begin{equation}\label{eq:eif}
		C \dfrac{dV}{dt} = -\varrho_L(V-V_L)+ \varrho_T\Delta_T \exp \bigg( \dfrac{V-V_T}{\Delta——T} \bigg)+I_{ion}.
	\end{equation}
In this equation, $C$ is the membrane capacitance,
$V_T$ is the membrane potential threshold,
$\Delta_T$ is the sharpness of action potential initiation, or slope factor,
$V_L$ is the leak reversal potential,
$\varrho$ is the conductance,
and $I_{ion}$ is input current. While $I_{ion}$ only represents synaptic current in Fourcaud's model, here we have extended the ionic current by summing up input current $I_{neib}$ from neighbors with connectivity, external noise current $I_{noise}$ that integrates i.i.d. Wiener process, and synaptic input $I_{syn}$ that incorporates the drifting term, serving as a bias. We have
		$I_{ion}=I_{neib}+ I_{noise}+I_{syn}$, where the term $I_{neib}$ takes identical form as leakage current in the first term of (\ref{eq:eif}), while $I_{syn}=\varrho_{syn}\Gamma(V-V_{syn})$ represents the slow voltage activated current with a gating variable $\Gamma=\Gamma(t)$. The term $\Gamma_\infty=\lim_{t\to\infty}\Gamma(t)$ can be used to describe the instantaneous activation.

Here, we alter the usually constant conductance $\varrho$, and change it to a function of $V$ and $t$. Multiplying $dt$ to both sides of the (\ref{eq:eif}), we now have
 	\begin{align}\label{eq:eif_xdt}
		C dV =& -\varrho_L(V-V_L){dt}+ \varrho_L\Delta_T \exp \bigg( \dfrac{V-V_T}{\Delta——T} \bigg){dt} \nonumber\\
		&+ \varrho_{syn}\Gamma(V,t)(V-V_{syn}){dt} + \varrho_{neib}(V-V_{noise}){dt} \nonumber\\
		&+ I_{noise}dt.
	\end{align}

It is clear that most terms in (\ref{eq:eif_xdt}) have very similar forms to those in (\ref{eq:voltage_ddm}). As most current terms do not need to be altered to fit in (\ref{eq:voltage_ddm}), the conductance term can change over time and become a function. Functions $\varrho$ and $\alpha$ are sometimes interchangeable. However, the exponential term can be tricky to work around, and we will talk about it in the later part after done discussing absorbing state.

Henceforth, for simplicity purpose, we call this EIF and DDM combined model as Exponential Decision Model, or EDM for short.

\section{Decision Making Dynamics}\label{sec:dm_dynam}

The commonly discussed decision making process is an adaptive behavior that makes use of a series of external input variables and then leads to an optimal or sub-optimal choice of action over other competing alternatives.

We begin by discussing this optimal decision rule. There are two thresholds $z_i$ in DDM model, with the same magnitude but different signs to represent different choices. In EDM, we consider this choice to be optimal if the threshold of correct choice is reached, or sub-optimal, if our expectation value $\mathbb{E}$ ends up with the same sign as correct threshold but with smaller magnitude.

Now we start solving the OU process described in (\ref{eq:voltage_ddm}).
\begin{lemma}\label{thm:sol_ou}
The solution of the collective decision making system in (\ref{eq:voltage_ddm}) is given by
	\begin{equation}\label{eq:xsol_ref}
		x = e^{\alpha(t)}\bigg(c+\int^t_{0}e^{-\alpha(t-\eta)}g(\eta)d\eta+\int^t_{0}e^{-\alpha(t-\eta)}\beta(\eta)dw_{\eta} \bigg)
	\end{equation}
which has a similar form as in \cite{arnold1974stochastic}, with the updated expectation
\begin{equation}\label{eq:xsol_mean}
	\mathbb{E}(x(t)) = \int^t_{0}e^{-\alpha(t-\eta)}g(s)ds.
\end{equation}

\end{lemma}

\begin{proof}
Let $\phi(t)$ be a fundamental solution matrix.

Also, let $Y$ be 
$$
c+\int^t_{0}\phi(\eta)^{-1}g(\eta)d\eta + \int^t_{0}\phi(\eta)^{-1}\beta(\eta)dw_{\eta}.
$$
Then $Y$ has the stochastic differential equation
$$
dY = \phi(t)^{-1} \big( g(t)dt+\beta(t)dw_t      \big),
$$
which implies
$$  
\begin{aligned}
x &= \phi(t) Y \\
&= \phi(t)\bigg(c+\int^t_{0}\phi(\eta)^{-1}g(\eta)d\eta + \int^t_{0}\phi(\eta)^{-1}\beta(\eta)dw_{\eta} \bigg).
\end{aligned}
$$
Furthermore, 
combining all above together, it is straightforward to conclude (\ref{eq:xsol_ref}).
\end{proof}

\subsection{Behaviors}\label{ssec:each_behav}

Here we consider a very special but common network system.
\begin{assumption}\label{brownagents}
For a high dimensional dynamical network system with $N$ agents described by (\ref{eq:voltage_ddm}), each unit integrates an inward stimulus $\alpha_i$ and receives signals $I_{neib}$ and noise $\beta_idw_i^j$ from local $j$ neighbors. The dynamics of $w^j$ follows Correlated Brownian motions, with standard correlation $col \in (-1,1)$.
\end{assumption}

In Reference \cite{tao_li}, the authors proposed a It\^{o} consensus S.D.E formalized by $N^2$-dimensional standard white noise, and expanded the right hand side terms in (\ref{eq:pureddm}) to a matrix form with graph theory, i.e., $\alpha(t)Lx(t)$. 
Different from their encoded gain function, we now assume that all the state information is available to others.
\begin{assumption}
For the considered network system with $N$ agents, the state of each agent is observable by others.
\end{assumption}

Extending (\ref{eq:voltage_ddm}), then we have the dynamic equation for each agent with neighbor's dynamics added
	\begin{multline}\label{eq:single_net_ddm}
		dx_i = \bigg( \sum_{j=1}^K \alpha(x_i(t),t)l_{ij}(y_{ji}-x_i(t))+ g(x_i(t),t) \bigg) dt \\ + \beta(x_i(t),t) dw,
	\end{multline} 
where $K$ is the total number of agents connecting the agent $i$,
$l_{ij}$ are elements in the Laplacian matrix $L$, $y_{ji}$ denotes the observed membrane potential of $j$th agent by the $i$th agent. 


In Reference \cite{girsanov}, Charalambos \textit{et al} have shown a method of achieving optimization problems under a reference probability measure by transferring continuous and discrete-time stochastic dynamic decision systems, via Girsanov's Measure Transformation. 
To this end, we have the following claim.

\begin{proposition}\label{prop:equi}
The collective stochastic dynamic decision system with common team optimality can be transformed to the equivalent static optimization problem with independent distributed sequences. And under the reference probability space, states and observations are independent Brownian motions.
\end{proposition}
 
Consider a series of $d$ inputs, $\textbf{X}_i(t)=[x_{i1}(t),...,x_{id}(t)]$, for $N$ agents, each of which has 2 states
\begin{equation}
  s_i =\begin{cases}
    1, & \text{if $t\in(t-\tau_{ref},t] $},\\
    0, & \text{otherwise}.
  \end{cases}
\end{equation}
The firing during $\tau_{ref}$ is sometimes called absolute refractory period.

Such series of inputs can be modeled as discrete decision making scheme, which we have applied the free-response paradigm upon, with adaptive prescribed time interval $\tau_{x}$ for each agent in the network system.
Equation (\ref{eq:single_net_ddm}) then becomes the following distributed protocol
	\begin{equation}\label{eq:measureable_ddm}
		d\textbf{X} = \bigg(L\cdot\alpha(\textbf{X}(t),t)\cdot\big(\textbf{Y}(t)-\textbf{X}(t)\big)+g(t)\bigg)dt + \beta(\textbf{X}(t),t) d\textbf{w}.
	\end{equation}

\subsection{Sampling with EIF}\label{ssec:sample}
Recall from Proposition \ref{prop:equi} that, since the Brownian motions are indpendent of all other team decisions, it opens up the possibility for Markov chain related methods. 
For efficiency and flexibility purpose, Markov Chain Monte Carlo (MCMC) has been applied in sampling the spiking network of neurons \cite{neural_sampling_2011}. 

In Reference \cite{richardson_2009}, Richardson has shown the equilibrium value of a slow driven voltage-activated current gating variable with the form
$ \tau_\Gamma \dfrac{d\Gamma}{dt} = \Gamma_{\infty} - \Gamma$, 
where $\tau_\Gamma(V)$ is an adaptive time constant characterized by different voltage values, and $\Gamma_{\infty}$ is the equilibrium value.	
Then we have $\Gamma=\Phi(\left(V_\Gamma-V\right)/\Delta_\Gamma)$, where $\Phi$ is the sigmoid function that digests a membrane potential function into a probability density function with range $[0,1]$.
	\begin{equation}\label{eq:gateequi}
		\Gamma_{\infty}=\dfrac{1}{1+e^{-(V-V_\Gamma)/\Delta_\Gamma}}.
	\end{equation}
Here, the equilibrium term $\Gamma_{\infty}$ holds a very similar form to Boltzmann probability distribution as in (\ref{eq:bmprobdist}). 

It is clear that (\ref{eq:gateequi}) is a slowly varying function, which can be proved simply by applying the definition.
Now we can use the property of the slow varying function to deal with the exponential term in (\ref{eq:eif_xdt}). 
The Karamata representation theorem is one of the mostly used property of slow varying functions that transfers a function into a general exponential form. In our case, $\Gamma_\infty$ is expressed as: $\Gamma_\infty=\exp  \left( \hslash(\Gamma) + \int_B^\Gamma \frac{\varepsilon(t)}{t} \,dt \right)   $, for some $B>0$, where $\hslash(\Gamma)$ is a bounded measurable function converging to a finite number, and $\varepsilon(t)$ is a bounded measurable function converging to $0$.

Here, since the exponential term is only a property of membrane potential increment that adds to nonlinearity, the potential accumulation of each agent does not affect collective network decision as much during the firing period $(t-\tau_{ref},t]$. And for the absolute refractory period, the neuron model is guaranteed not to fire. For such a piece wise continuous function, if we only consider the time interval to be one firing, then the variable $V$ can be bounded. Taking out the exponential term in (\ref{eq:eif_xdt}), we have
$$ 
\dfrac{CdV_{e}}{\varrho_T(V_e,t)\Delta_T dt} =  \exp \bigg( \dfrac{V_e-V_T}{\Delta——T} \bigg),
$$

$$
\begin{aligned}
\lim_{t\to t_{end}}\exp\left(\dfrac{V_e-V_T}{\Delta_T}\right)&=\lim_{t\to t_{end}}\exp(\hslash(V_e)) \\  
&=\exp\left(\dfrac{V_{end}-V_T}{\Delta_T}\right) = \mathcal{C} \\
&= \tau_\Gamma \dfrac{d\Gamma}{dt}+\Gamma.
\end{aligned}
$$
where $t_{end}$ is the end of the refractory time,  
$V_{end}$ is the membrane potenial at end of the refractory period, and $V_e$ is the membrane potential incremented by exponential term only. Since $\hslash$ converges to a finite number and $\varepsilon(t)$ converges to 0, the limit of the exponential term converges to a constant $\mathcal{C}$ during the refractory period.

It can be thought of as entering an absorbing state where its behavior at infinity is very similar to the behavior of converging to infinity. Therefore, at the absolute refractory period, we can safely ignore this exponential term in (\ref{eq:eif_xdt}), and treat it just as other linear terms.

\begin{figure}[htbp]
  \centering
      \includegraphics[width=0.48\textwidth]{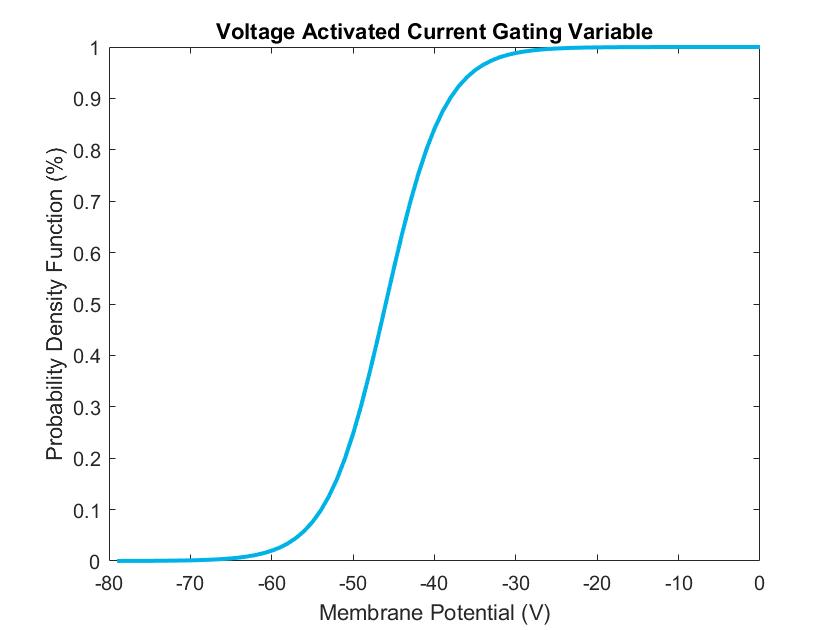}
  \caption{Voltage Activated Current Gating Variable}{\label{fig:boltzmann}}
\end{figure}

For simulation, we have used the same parameters as the work done by Barranca \textit{et al} \cite{Barranca2014}.
As shown in Fig. \ref{fig:boltzmann}, this monotonically increasing activation function has a sigmoidal general form, with an upper bound 1 and lower bound 0. In some cases, as mentioned above, we may set the initial condition to be $0$ mV to eliminate the afterhyperpolarization. For our simulation, we have removed such constraints to show that the dynamics of the membrane potential can be bounded by the activation function with arbitrary initial conditions.
This is because the higher threshold in the EIF model always requires higher voltage to be breached, and the logistic activation function describes such positive correlation well enough.

In fact, very similar dynamics of activation function has been capitalized in the neural sampling frame work \cite{richardson_2009,neural_sampling_2011,sto_infere}. Thus, our collective decision making model can be thought of as a network consisting of $N$ agents (or neurons) sampling from a probability distribution $p$ using the stochastic dynamics carried from the Drift-Diffusion model.

\begin{proposition}
The firing activity of the generalized EIF model, which represents each agent in a collective DDM network system, follows a Markov chain process. 
\end{proposition}

With the information-coded signal from each DDM agent of the system, in other words, the firing information within the time interval $(t-\tau_{ref},t]$, the neural sampling follows conditional probability distribution, and most of time is Boltzmann distributions.

For each agent, we treat the collective behavior of connected nodes as a accelerator/damper. For instance, if the agent is surrounded by nodes with higher membrane potential, it receives more current than normal drifting, and vice versa.

\begin{flushleft}
$ p(s_i=1|X_i(t-1))=  $
\end{flushleft}
\begin{equation}\label{eq:prob_each}
\dfrac{X_i(t-1)+g(t)+\sum_{j=1}^K\alpha(t)l_{ij}\big(y_{ji}-x_i(t)\big)}{V_T+\dfrac{\sum_{j=1}^K\alpha(t)l_{ji}\big(y_{ij}-x_j(t)\big)}{K}}.
\end{equation}

\section{Collective Behavior}\label{sec:coll_beh}

For most multi-agent dynamic network systems, it is common that new communication links can be established over two agents with no previous connection. In the previous sections, we have assumed that the connection is known at each state. Now we define the coupling and connection behavior among the agents in the system.

\subsection{Coupling and Connectivity}

\begin{assumption}\label{asp:coup_neib}
For the system described in this paper, agent $i$ forms at most $K$ outward links randomly at $t=0$. 
When $s_i=1$, Agent $i$ tries to establish new connections with new neighbors, for instance, connecting to Agent $j$ with the coupling probability $\mathcal{P}_{ij}$ depends on the voltage difference. When successful, the equal number of previous connections are lost according to a decoupling probability function $\mathcal{Q}$. When $s_i=0$, Agent $i$ will not actively modify its neighboring connections.
\end{assumption}

It is worth pointing out that both $\mathcal{P}$ and $\mathcal{Q}$ follow a sigmoid (or reverse sigmoid) relation $\Phi$ (or $1-\Phi$).
For function $\mathcal{P}$, the greater the difference in membrane potentials has, the higher the probability is. The situation is reversed for function $\mathcal{Q}$, the greater the difference in membrane potentials has, the lower the probability is.   

Then we have the following equations
$$
\mathcal{P}_{ij} = \dfrac{X_j-X_i}{V_{Ti}+V_{Tj}}\dfrac{K-\mathcal{K}_i}{K},
$$
$$
\mathcal{Q}_{ij} = \left(1-\dfrac{X_i-X_j}{V_{Ti}+V_{Tj}}\right)\dfrac{\mathcal{K}_i}{K},
$$
where $\mathcal{K}_i$ is the number of current established connections of node $i$. In the case of probability values that are less than zero or greater than one, we simply set them to $0$ and $1$ respectively for mean field analysis. The negative probability is also provided in Fig. \ref{fig:globalbranch} . Different from the common equation for branching probability, our process is not a tree-like process and the maximal connectivity is defined by $K$. Therefore they are Markovian processes with respect to the number of connected nodes.

\begin{figure}[htbp]
  \centering
      \includegraphics[width=0.48\textwidth]{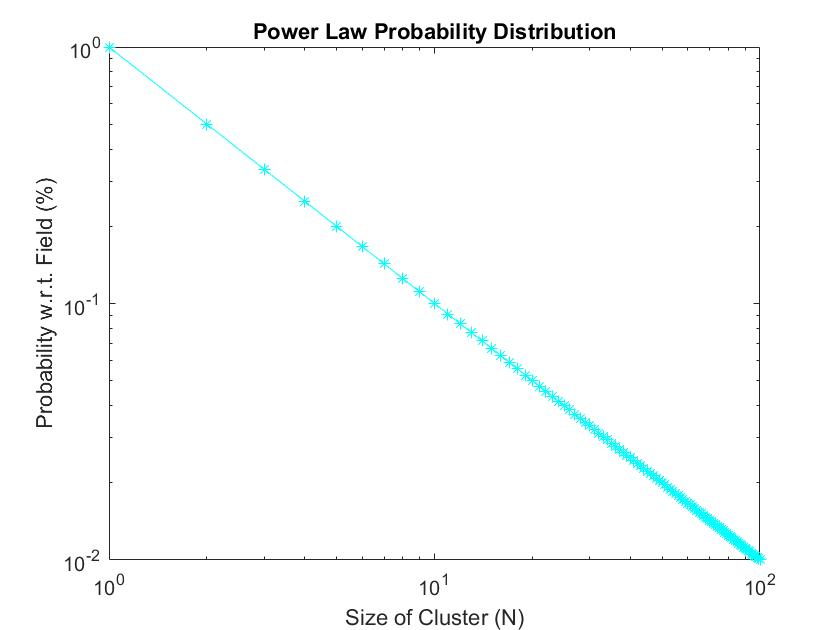}
  \caption{Power Law Probability Distribution}{\label{fig:powerlaw}}
\end{figure}

\subsection{Mean Field Analysis}

In the mean field analysis, the individual drifting variable usually follows a distribution with the average. In our case, we assume that the expanded term follows the average distribution of $g_i/W$, where $W$ can be considered as the sum of the total drift. 
Here $g_{i}$ is still a gain function representing the external input exerted by neighbors.
However, due to the free-response property, $g_{i}$ is essentially a piece-wise continuous gain function representing synaptic strength.

Recalling the term ``active site" described in Section {\ref{ssec:socintro}}, the local active site for each agent can be thought as the local field: 
$\mathcal{F}_i=-\sum_j^Na_{ji}g_{j}/W$,
where $a_{ji}$ is the elements in an adjacency matrix, then the global field has the average 
$\bar{\mathcal{F}}=\sum^N_i  \left(\mathcal{F}_i/K\right)/N $.

Therefore, with regard to a sequence of $d$ number of inputs, the probability of a single DDM having the state $s_i=1$, represented by a single generalized EIF neuron, has approximately the following probability with respect to the field. 
\begin{equation}\label{eq:eifboltz}
\begin{aligned}
p_i(s) = &S^{-1} \exp\bigg(-\gamma \bigg(\sum_{j=1}^{N \setminus K}\bar{\mathcal{F}}\mathcal{P}_{ij}-\\
&~~~~~~~~~~~\sum_{j=1}^K\mathcal{F}_j \mathcal{Q}_{ij}+\sum_0^d b_i \bigg)\bigg),
\end{aligned}
\end{equation}
where $S$ is some partition function, $b_i$ is the bias term that supports the correct choice, $\gamma$ is a thermodynamic beta in Boltzmann factor with the form $\gamma = 1/\left(k_B \mathcal{F}_i\right)$, and $k_B$ is a Boltzmann constant.

Putting aside the coupling strength, the mean activity of this network is measured as 
$\sum^N_{i=1}s_i/N$.
However, for such a spiking neural model, it is not always practical to have normally distributed probability density function, and in fact, most of these processes are stochastic with certain thresholds or even highly constrained. That being said, we would have a stochastic It\^{o} based integral for probability density for the mean activity: 
$\mathcal{H}=\int^{\infty}_{-\infty}\Phi(X)p(s_i|X)dX$.

\section{Convergence Analysis}\label{sec:converge}

There are two main evidences for systems presenting the SOC behavior \cite{bak_1996,Watkins2016}, the first is the power law distribution, and the second is the critical dynamics, which we consider as converging to absorbing states in the EDM model. 
In this section, we expand the results from Sections \ref{sec:dm_dynam} and \ref{sec:coll_beh}, and examine the global convergence behavior of the collective EDM model. We then provide both pieces of evidence to show the EDM system has the SOC behavior.

\subsection{Global Criticality from Local Dynamics}\label{ssec:globfloc}
Recall that in SOC, active nodes trigger self avalanches when a threshold is reached, and update the information of all connecting active nodes.
Also, nodes in nearby inactive sites will be communicated to establish more connections if needed.

Moreover, each local active site has their own dynamics of reaching out to other active or inactive sites. 
Inspired by the work of Harris on the theory of branching processes \cite{harris2002theory}, we use a branching parameter $\sigma$ that captures the subsequent activity of connectivity triggering or dying-out \cite{soc_branching}.
The \textit{local branching ratio} and \textit{global branching ratio} have the form

$$\sigma_j(t)= \sum^K_{i=1} \mathcal{P}_{ij}(t), \quad\tilde{\sigma}(t)=\dfrac{1}{N-1}\sum^N_{j=1}\sigma_j(t),$$
respectively. As discussed in \cite{harris2002theory,soc_branching}, the system exhibits criticality at $\sigma = 1$, and is sub-critical (super-critical) for $\sigma < 1$ ($\sigma > 1$).

\begin{figure}
    \centering
    \begin{subfigure}[thpb]{0.48\textwidth}
        \includegraphics[width=\textwidth]{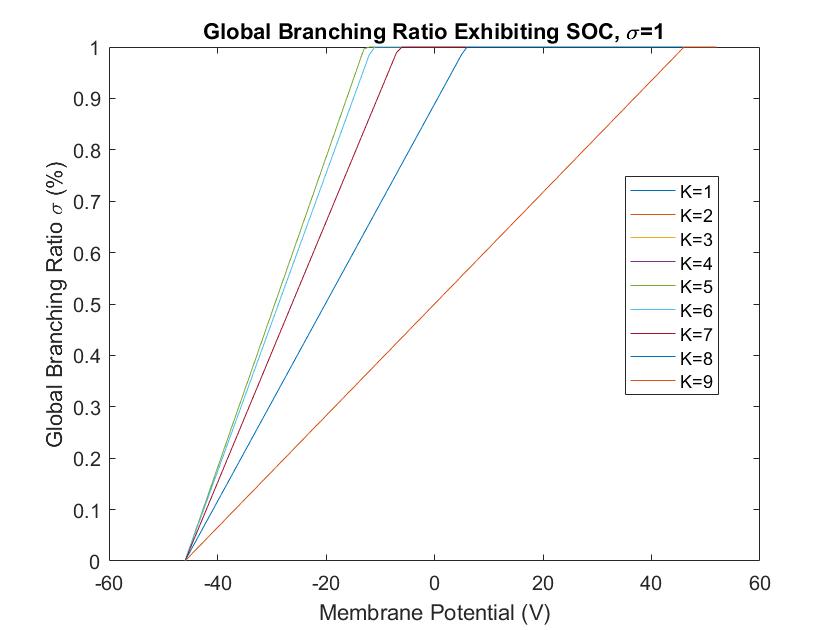}
        \caption{}
        \label{fig:gbmp}
    \end{subfigure}
    ~ 
    \begin{subfigure}[thpb]{0.48\textwidth}
        \includegraphics[width=\textwidth]{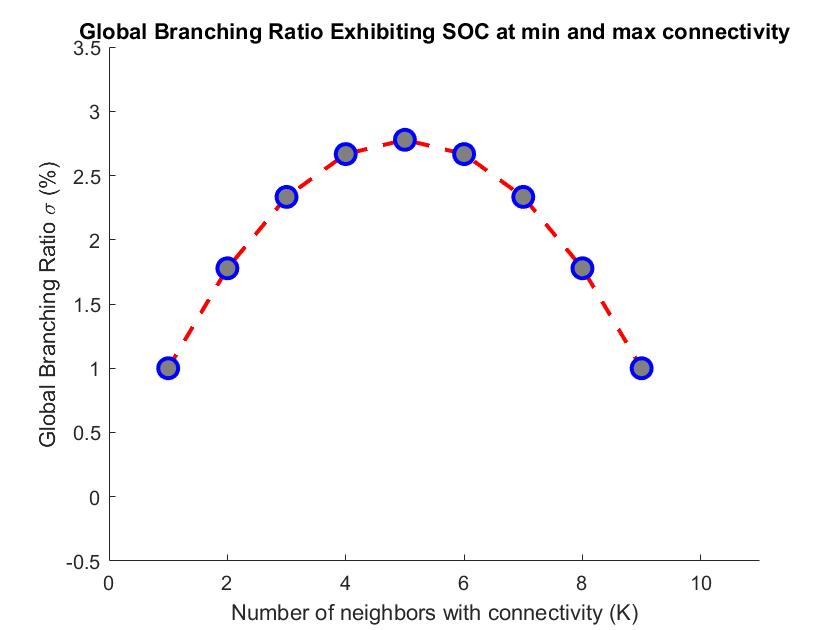}
        \caption{}
        \label{fig:gbnn}
    \end{subfigure}
    \caption{In Fig. (\ref{fig:gbmp}), as the firing probability increases with the membrane potential, $\sigma$ converges to 1 with proper connectivity constraints. Cases with different numbers of active and connected neighbors are shown in a network system with $N=10$. Fig (\ref{fig:gbnn}) removes the constraint that cumulation of local branching ratio in each iteration caps at 1. It is clear that as the number of connected neighbors increases, the network system enters the active/loading phase first, and then evolves to the dissipation/absorbing phase. The system clearly shows SOC dynamics, that is $\sigma = 1$ at both minimal and maximal connectivity.}\label{fig:globalbranch}
\end{figure}

In our simulation, sub-critical dynamics are characterized by low potential and rapidly decaying agent neuron firing distributions, while super-critical dynamics are characterized by high potential and slowly decaying firing activities. 
Critical dynamics are characterized by firing activity that follows power-law distributions.

As shown in Fig. \ref{fig:powerlaw}, it can be easily recognized that the collective behaviors of the firing density in the EDM model follow a power law distribution, 
$\mathcal{D}(z) \sim z^{\Lambda}$ with different cluster size $z$ and scaling factor $\Lambda<0$, which is a primary evidence supporting the SOC behavior \cite{Watkins2016}.

\subsection{Absorbing States}\label{ssec:absorb}

Recalling the avalanches described in self-organized criticality, the system keeps tuning itself to one of many meta-stable states, which commonly have lifespans shorter than ground state and longer than excited states \cite{bak_1996}. And without further inputs, the distribution reaches meta-stability, a very special energy well that is able to temporarily trap the system for a limited number of states. 

This can be modeled as Boltzmann distribution with global energy level in a simulated annealing system from any initial conditions. 
With the results from (\ref{eq:prob_each}) and (\ref{eq:eifboltz}), as well as the exponential property of the generalized EIF system, the Lyapunov based semistability \cite{HHB:TAC:2008} can be achieved with great potential, but due to the length restriction, this concept will not be discussed in this paper. Nevertheless, we are going to show that the system described in this paper converges to some absorbing states.

In mean field theory, as being discussed in \cite{soc_branching}, absorbing state becomes unstable when the probability of a node creating connection with neighbors is greater than $1/2$. In our case, this can be thought as the coupling probability  
$\mathcal{P}>\mathcal{P}_{critical} = 1/2$.

\begin{lemma}\label{lem:multi_absorb}
The attractor of the system is a set of discrete states.
\end{lemma}

\begin{proof}
If a non-conserving system, such as the drift and diffusion based EDM model described in this paper, has shown a temporary stable configuration after the avalanches, then the system is at least at a critical point. The critical and super-critical session are usually slow driving \cite{soc_branching}. So there must be a drift load and a diffusion dissipation fluctuating to keep all the nodes in the system from either forming active sites, or staying quiescent completely.

Thus, if the system presents the thermodynamic behavior as Boltzmann distribution with simulated annealing, there can exist infinite numbers of infinitesimally varied absorbing states in thermodynamic limit. As for a finite number of total states, this absorbing phase becomes a set of discrete states.
\end{proof}

This obeys another property of SOC, that is, the dynamical system with a critical point as an attractor, is able to keep itself at the critical point between two phases, which in our case, are the active phase and the absorbing phase.

\begin{proposition}
Meta-stable states generally hold more energy than the ground states, and less energy than the excited states. Therefore, SOC is the process of the EDM model losing global energy, and falling into a certain set of absorbing states, regardless of guaranteed stability.
\end{proposition}

Each agent in the EDM model is essentially drift diffusion terms taking input variable from EIF markup. It is clear that when the individual thresholds are reached, the agents will initiate the spike and send information-coded signals (current) to connected nodes or nodes with great probability of establishing connectivity. 
The exited states usually carry higher membrane potential than the incremental states, and then the fired node resets its membrane potential to $V_R$ and enters a refractory period. Therefore, the states during this absolute refractory period $\tau_{ref}$ can be considered as comparable meta-stable states that trap the dynamics of each node for $\tau_{ref}$. And since we have proved Lemma \ref{lem:multi_absorb}, the system self-organizes itself, instead of fine-tuning, to a small region of absorbing criticality, or in another word, metastability.

Since the system is slow varying at the absorbing state, the fundamental solution is independent of $t$, that is, $\alpha(X(t),t)) = \alpha(X)$. And (\ref{eq:xsol_ref}) becomes
	\begin{equation}\label{eq:xsol_ref_inv}
		X = e^{\alpha(t-\tau_{ref})}c+\int^t_{t_0}e^{-L\alpha(t-\eta)} \big(g(\eta)d\eta+\beta(\eta)dw_{\eta} \big) .
	\end{equation}

At such slow varying, time independent absorbing states, it is natural to assume that the dimension of $\beta$ is only 1. Using the corollary in \cite{arnold1974stochastic},
$ \phi(t) = \exp\bigg( \int^t_{(t_0)} \alpha(\eta)d\eta \bigg)$.

Then we can further turn Theorem \ref{thm:sol_ou} to be
\begin{equation}\label{eq:sol_ou_gen}
\begin{aligned}
&X(t) = \exp\bigg( \int^t_{(t_0)} \alpha(\eta)d\eta \bigg)\bigg(c + \\
&~~~~\int^t_{t_0}\exp \bigg( \int_{t_0}^s A(u)du \bigg) \big(g(\eta)d\eta+\beta(\eta)dw_{\eta} \big) \bigg) .
\end{aligned}
\end{equation}

The threshold presents the local rigidity level. Also, the convergence dynamics of these absorbing states show Boltzmann distribution as well, i.e.,
$	
		P_i(V_x|s=0)= e^{-X_i/k_B\bar{\mathcal{F}}}/\left(\sum_{j=1}^N e^{-X_i/k_B\bar{\mathcal{F}}}\right).
$	
	
As shown in Fig. \ref{fig:globalbranch}, the branching pattern of the EDM model across multiple fires follows the SOC behavior and will eventually evolve to a certain set of absorbing states, known as the recurrence sets.
Also, these stationary distributions of network states can be convergent from any initial state.

Therefore, without the presence of a proper controller, the system fine-tunes itself, and then converges to  
$   
\bar{u} = \lim_{t \to \infty} \sum_0^{t-1} \mathbb{E}(X(t))/t.
$

\section{Conclusion and Future Work}\label{sec:conclusions}

\subsection{Conclusion}

In this work, we have proposed a collective decision making model with a specific type of spiking neurons, exponential integrate-and-fire (EIF). Our method is based on the well-known Two-Alternative Forced Choice Task solver--Drift-Diffusion model. We recognize that DDM and EIF share very common terms in their dynamic equations, and the exponential term can be ignored during the absolute refractory period. We have derived the probability of each agent's firing based on a Markov chain conditional premise. Then the mean field theory is used to approximate the global criticality from local dynamics. 

Both analytically and experimentally, we have found out that the global branching ratio follows a power law distribution and the EDM system eventually evolves to a set of absorbing states, which are two main evidences suggesting the Self-Organized Criticality behavior. The activation function follows the Boltzmann state probability and the convergence dynamics of absorbing states follow Boltzmann distribution as well.

\subsection{Future Work}

At this point, we have set up a detailed model that is ready to be expanded from different aspects. For instance, since theoretical and experimental studies have demonstrated that critical systems are often optimizing computational capability, it is promising to suggest that the system with the SOC behavior is both robust and flexible to ensure homeostatic stability.

In fact, due to the nature of absorbing states and criticality property, any initial conditions of a spiking network decision making system can converge and/or fluctuate around a set of states, potentially semistability \cite{HHB:TAC:2008}. Therefore, the convergence property of such a model can be useful for fault pre-screening and is in a way, robust to quantified uncertainties.

Besides the theoretical analysis, there are potential applications as well. Bak has demonstrated that both traffic dynamics and brain dynamics exhibit similar criticality \cite{bak_1996}. This opens up the way to extend our model to formulate real-world applications. With specific problem solving scenarios, it is natural to extend the proposed multi-agent spiking neural decision making system to a parallel distributed process, and eventually leads to a neuromorphic chip development.

Since the nature of SOC is to organize the system between two phases, the fast switching capability is useful for sensitivity analysis. Furthermore, we will also look into fully stochastic neurons such as \textit{Galves-L\"{o}cherbach} (GL) model to incorporate different activation profiles and to increase the candidate choices as well.  

\balance

\bibliographystyle{IEEEtran}
\bibliography{IEEEabrv,library}

\end{document}